\documentclass[12pt]{article}
\usepackage{amsmath}
\usepackage{amsthm}
\usepackage{amssymb}
\newcommand{\bu}{\boldsymbol{u}}

\newtheorem{thm}{Theorem}

\newtheorem{pro}{Proposition}
\date{January 17, 2019}
\begin{document}

\title{\bf Integrable (3+1)-dimensional system\\ with an algebraic Lax pair}
\author{A. Sergyeyev\\[2mm] Mathematical Institute, Silesian University in Opava,\\ Na
Rybn\'\i{}\v{c}ku 1, 746 01 Opava,~Czech Republic\\[2mm]
E-mail: \texttt{Artur.Sergyeyev@math.slu.cz}}

\maketitle
\begin{abstract}
\protect\vspace*{-5mm}
  We present a first example of an integrable (3+1)-dimensional dispersionless system with nonisospectral Lax pair involving algebraic, rather than rational, dependence on the spectral parameter, thus showing that the class of integrable (3+1)-dimensional dispersionless systems with nonisospectral Lax pairs is significantly more diverse than it appeared before. The Lax pair in question is of the type
recently introduced in [A. Sergyeyev, Lett. Math. Phys. 108 (2018), no. 2, 359-376, arXiv:1401.2122].\looseness=-1
\end{abstract}

%\begin{keyword}
%nonisospectral Lax pairs \sep (3+1)-dimensional integrable systems \sep dispersionless systems
%\MSC[2010] 37K05\sep 37K10
%\end{keyword}
%\maketitle

%\linenumbers

\section*{Introduction}

Integrable systems are well known to play an important role in modern mathematics and physics, cf.\ e.g.\ \cite{as,ca,df,d,dfk,fer,hkbp,kod,kvv,mw,o,as-ro,scg,se,vps,w,yan,z}. According to Einstein's general relativity our spacetime is four-dimensional, so the search for (3+1)-dimensional integrable systems, that is, integrable partial differential systems in four independent variables, is quite naturally of particular significance, see e.g.\ \cite{d,mw,scg,w}; also cf.\ e.g.\ \cite{t}
%and references therein
on pecularities of 4-manifolds.
As discussed e.g.\ in \cite{scg}, most of integrable (3+1)-dimensional systems known to date are dispersionless, which, roughly speaking, means that they can be written as first-order quasilinear homogeneous systems.
\looseness=-1

While a fairly large number of (3+1)-dimensional integrable systems is presently known, see e.g.\ \cite{mw,as-ro,scg} and references therein, the breadth of their class yet remains to be fully understood. In 2+1 dimensions there are many examples of dispersionless integrable systems whose
nonisospectral  Lax pairs involve highly sophisticated dependence on the spectral parameter, including, for example, Weierstrass $\wp$-functions, cf.\ e.g.\ \cite{fer}, but, to the best of our knowledge, for all previously known examples of (3+1)-dimensional integrable dispersionless systems in finitely many dependent variables with nonisospectral Lax pairs the latter are  polynomial or rational in the spectral parameter, which begets the question of whether more sophisticated nonisospectral Lax pairs could exist in 3+1 dimensions.\looseness=-1

We answer this question in the affirmative by presenting in Section~\ref{nis} below system (\ref{sys}) which is, as far as the present author is aware, the first example of an integrable (3+1)-dimensional dispersionless system in finitely many dependent variables with a nonisospectral Lax pair being  algebraic, rather than merely rational, in the spectral parameter, which shows that nonisospectral dispersionless Lax pairs in 3+1 dimensions are significantly more diverse than it appeared before. The example in question is found within the framework of a new systematic construction for (3+1)-dimensional integrable systems related to contact geometry, see  \cite{scg}, the follow-up papers \cite{bls, snd}, and a brief review in Section~\ref{pre} below.
\looseness=-1
\section{Preliminaries}\label{pre}

Recall that a {\em dispersionless} or {\em hydrodynamic-type} system in four independent variables $x,y,z,t$ is, cf.\ e.g.\ \cite{dfk,fer,scg, se} and references therein, a system that can be written in the form
\begin{equation}\label{sys-gen}
A_0(\bu)\bu_t+A_1(\bu)\bu_x+A_2(\bu)\bu_y+A_3(\bu)\bu_z=0,
\end{equation}
where $\boldsymbol{u}=(u_1,\dots,u_N)^{\mathrm T}$ is an $N$-component vector of unknown functions of $x,y,z,t$,
$A_i$ are $M\times N$-matrix-valued functions of $\bu$, $M$ and $N$ are nonzero nonnegative integers such that $M\geqslant N$, and the superscript ${\mathrm T}$ indicates the transposed matrix.

Here and below all functions are assumed sufficiently smooth for all computations to make sense. This can be readily formalized using the language of differential algebra, cf.\ e.g.\ \cite{dsk,scg}, and references therein.\looseness=-1

There exist \cite{scg} infinitely many {\bf integrable} dispersionless
systems of the general form (\ref{sys-gen}) that
%possess,
admit, for suitable
Lax
functions $f=f(p,\bu)$ and $g=g(p,\bu)$, %linear
nonisospectral Lax pairs
of the form introduced in \cite{scg} and intimately related to contact geometry,\looseness=-1
\begin{equation}\label{clp-gen}
\chi_y=X_f(\chi),\quad \chi_t=X_g(\chi),
\end{equation}
where $\chi=\chi(x,y,z,t,p)$, and $p$ is the so-called variable spectral parameter (note that $\bu_p\equiv 0$).

For any $h=h(p,\bu)$ the operator $X_h$ is defined as
\begin{equation}\label{cvf}
X_h=h_p\partial_x+(ph_z-h_x)\partial_p+(h-p h_p)\partial_z
\end{equation}
and formally looks exactly like the contact vector field with a contact hamiltonian $h$ on a contact 3-manifold with local coordinates $x,z,p$ and contact one-form $dz+pdx$; cf.\ e.g.\ \cite{bl, br}
and references therein for more details on contact geometry.

In particular, for any natural $m$ and $n$ the pairs of Lax functions
\[
f=p^{n+1}+\sum\limits_{j=0}^n v_j p^j,\quad g=p^{m+1}+\frac{m}{n}v_{n} p^m+\sum\limits_{k=0}^{m-1}w_k p^k
\]
and
\[
f=\sum\limits_{j=1}^m a_i/(v_i-p),\quad g=\sum\limits_{k=1}^n b_k/(w_k-p)
\]
yield \cite{scg, snd} (3+1)-dimensional integrable systems of the general form (\ref{sys-gen}) with $M=N$ for $\bu=(v_0,\dots,v_n,w_0,\dots,\allowbreak w_{m-1})^\mathrm{T}$ and $\bu=(a_1,\dots,a_m, b_1,\dots,b_n, v_1,\dots,v_m,w_1,\dots,w_n)^\mathrm{T}$ respectively.

Notice \cite{scg} that if $\bu_z=0$ and $\chi_z=0$ then the Lax pairs (\ref{clp-gen}) boil down to well-known (2+1)-dimensional Lax pairs involving Hamiltonian (rather than contact) vector fields, cf.\ e.g.\ \cite{fer, z} and references therein for the associated (2+1)-dimensional integrable systems.

\begin{pro}[\cite{scg}]\label{pro} A system (\ref{sys-gen}) admits a linear Lax pair of the form (\ref{clp-gen}) if and only if it admits a nonlinear Lax pair for $\psi=\psi(x,y,z,t)$ of the form
\begin{equation}\label{nlp-gen}
\psi_y=\psi_z f(\psi_x/\psi_z,\bu),\quad \psi_t=\psi_z g(\psi_x/\psi_z,\bu)
\end{equation}
with the same functions $f$ and $g$ as in (\ref{clp-gen}). %Here $\psi=\psi(x,y,z,t)$.
\end{pro}

In closing note that systems of the form (\ref{nlp-gen}) belong to a broader class of %more general
multitime Hamilton--Jacobi systems, cf.\ e.g.\ \cite{lr} for more details on the latter.
\section{New integrable system with algebraic Lax pair}\label{nis}
Consider the following (3+1)-dimensional evolutionary system for %with
%$N=6$ and
$\bu=(a,b,r,s,u,v,w)^{\mathrm T}$:
\begin{equation}\label{sys}
\hspace*{-7mm}
\begin{array}{rcl}
a_t&=&\displaystyle\frac{1}{r^2-2 r s a + 2 s^2 b}\biggl((4 w (r a - s b) - v r) a_x + r a_y\\[4mm]
&& + \bigl( 2 w \left(2 a(r a - s b) - r b\right) - u r\bigr) a_z\\[2mm]
&&+ (v s - 2 w r) b_x - s b_y + (2 w (s b -r a) + u s) b_z\\[2mm]
&&+(r-s a) u_x + (r a - 2 s b) u_z + (2 s b-r a) v_x\\[2mm]
&&+ 2(a (s b- ra) + r b) v_z\\[2mm]
&& + 2(a (a r- s b)-rb) w_x + 2\left(2 a^2 (a r - s b) - 3 r a b + 2 s b^2\right) w_z\biggr),\\[5mm]
b_t&=&\displaystyle\frac{1}{r^2-2 r s a + 2 s^2 b}\biggl(2 (2 w r-v s) b a_x + 2 s b a_y\\[4mm]
&&+ 2 (2 w (r a -  s b) - u s) b a_z\\[2mm]
&& + (2 s(v a - 2 w b) - v r) b_x + (r-2 s a) b_y\\[2mm]
&& + (2 (u s a - w r b) - u r) b_z\\[2mm]
&&+(2 s (b-a^2) + r a) u_x + 2 (r-s a) b u_z\\[2mm]
&& - 2 (r-s a) b v_x - 2 (r a - 2 s b) b v_z\\[2mm]
&& + 2 (r a - 2 s b) b w_x + 4 (a (r a - s b) - r b) b w_z
\biggr),\\[5mm]
r_t&=&\displaystyle\frac{1}{r^2-2 r s a + 2 s^2 b}\biggl((v s - 2 w r) r a_x - r s a_y\\[4mm]
&&- (2 w (r a - s b) - u s) r a_z\\[2mm]
&&+ (2 w r-v s) s b_x + s^2 b_y + (w r^2-u s^2) b_z\\[2mm]
&&+(s a - r) s u_x + (2 s b-r a) s u_z + (r-s a) r v_x + (r a - 2 s b) r v_z\\[2mm]
&&+ (2 s b-r a) r w_x - 2 ( a(r a - s b) - r b) r w_z\biggr),\\[5mm]
s_t&=&w_x + a w_z + w a_z,\\
u_t&=&a r_x + 2 b r_z - s b_x,\\
v_t&=&r_x + a r_z + a s_x + 2 b s_z - s a_x + s b_z,\\
w_t&=&s_x + a s_z + s a_z.
\end{array}
%\vspace{-1mm}
\end{equation}

\begin{thm}\label{pr1}
The (3+1)-dimensional seven-component evolutionary system (\ref{sys})
is integrable since it admits a
%linear nonisospectral
Lax pair (\ref{clp-gen}) with
algebraic Lax functions $f$ and $g$ given by
\begin{equation}\label{fgp}
\begin{array}{rcl}
f&=&u+v p+w p^2+(r+s p) \sqrt{p^2+2 a p+2 b},\\[2mm]
g&=&\sqrt{p^2+2 a p+2 b},
\end{array}
\end{equation}
that is,
\begin{equation}\label{clp}
\begin{array}{rcl}
\chi_y&=&\displaystyle\frac{1}{g}\biggl(\bigl(2 s p^2+(r+3 s a+2 w g)p+r a+v g+2 s b\bigr)\chi_x\\[5mm]
&&+\bigl(-s p^3-(w g+s a) p^2+p r a+2 r b+u g\bigr)\chi_z\\[2mm]
&&+\bigl(s_z p^4+(2 a s_z+s a_z+r_z+g w_z-s_x)p^3\\[2mm]
&&+((v_z-w_x)g+2 b s_z+r a_z+s b_z-2 a s_x-s a_x-r_x+2 a r_z)p^2\\[2mm]
&&+((u_z-v_x)g+r b_z-r a_x-s b_x-2 b s_x-2 r_x a+2 b r_z)p\\[2mm]
&&-r b_x-g u_x-2 b r_x\bigr)\chi_p\biggr),\\[5mm]
\chi_t&=&
\displaystyle\frac{1}{g}\biggl((p+a)\chi_x+(a p+ 2b)\chi_z+(a_z p^2+p (b_z-a_x)-b_x)\chi_p\biggr),
\end{array}
\end{equation}
and a nonlinear Lax pair of the form (\ref{nlp-gen}) with $f$ and $g$ given by (\ref{fgp}), that is,
\begin{equation}\label{nlp}
\hspace*{-3mm}
\begin{array}{rcl}
\psi_y&=&u\psi_z +v\psi_x+w\psi_x^2/\psi_z
%\\[2mm]&&
+(r\psi_z+s\psi_x) \sqrt{(\psi_x/\psi_z)^2+2 a\psi_x/\psi_z+2 b},\\[3mm]
\psi_t&=&\psi_z\sqrt{(\psi_x/\psi_z)^2+2 a\psi_x/\psi_z+2 b};
\end{array}\hspace{-2mm}
\end{equation}
both of the Lax pairs (\ref{clp}) and (\ref{nlp}) are expressed in terms of algebraic, rather than rational, functions.
\end{thm}

%{\em Proof.}
\begin{proof}
By definition, proving that (\ref{sys}) admits (\ref{clp}) and (\ref{nlp}) amounts to proving that both (\ref{clp}) and (\ref{nlp}) are compatible by virtue of (\ref{sys}). %what id

First of all, by Proposition~\ref{pro} compatibility of (\ref{clp}) by virtue of (\ref{sys}) implies the same for (\ref{nlp}).

Next, by Proposition~1 of \cite{scg}, in order to prove that (\ref{clp}) is compatible by virtue of (\ref{sys}), it suffices to show that for $f$ and $g$ given by (\ref{fgp}) the system (\ref{sys}) implies
the equation
\begin{equation}\label{zcr}
{f_t-g_y+\{f,g\}}=0
\end{equation}
where the contact bracket $\{,\}$ is defined \cite{scg} as
\begin{equation}\label{cbr}
\{f,g\}=f_p g_x-f_x g_p-p(f_p g_z -g_p f_z)+f g_z- g f_z.
\end{equation}

Substituting (\ref{fgp}) into (\ref{zcr}) we readily see that this is indeed the case, and the result follows.
\end{proof}

\section{Conclusions}

We have presented above a new integrable (3+1)-dimensional dispersionless evolutionary system (\ref{sys}) whose linear Lax pair (\ref{clp}) and nonlinear Lax pair (\ref{nlp}) are expressed in terms of algebraic rather than rational functions.

To the best of our knowledge, this is the first example of an integrable (3+1)-dimensional dispersionless system in finitely many dependent variables admitting a nonisospectral Lax pair which is {\em not} rational in the variable spectral parameter. Indeed, as far as the present author is able to tell, for all previously known examples of integrable dispersionless (3+1)-dimensional systems in finitely many dependent variables with nonisospectral Lax pairs, including e.g.\ the equations for (anti-)self-dual four-dimensional conformal structures \cite{dfk}, the Dunajski equation, and systems (15), (17), (38), and (40) from \cite{scg}, %etc.,
their Lax pairs
have at most rational dependence on the variable spectral parameter.\looseness=-1 %, cf.\ e.g.\ \cite{dfk, scg} and references therein.\looseness=-1

We conjecture that (\ref{sys}) has no nontrivial linear or nonlinear Lax pairs written in terms of rational (rather than algebraic) functions.

In closing note that it could be of interest to study symmetries, conservation laws, Hamiltonian operators, and other related structures for (\ref{sys}) in spirit of \cite{kvv,o}.

\section*{Acknowledgments}
This research was supported in part by the Ministry of Education,
Youth and Sports of the Czech Republic (M\v{S}MT \v{C}R) under RVO
funding for I\v{C}47813059, and by the Grant Agency of the Czech
Republic (GA \v{C}R) under grant P201/12/G028. The
computations in the paper were mostly performed using the computer algebra package {\em
Jets} \cite{jets} for Maple\circledR.\looseness=-1

I am pleased to thank E. Ferapontov, I.L. Freire, B. Kruglikov, H.V. L\^e, R.O. Popovych, and R. Vitolo for stimulating discussions,  C. Taubes for sending me a copy of \cite{t}, and the anonymous referees for useful suggestions.\looseness=-1

\protect

\vspace{-4mm}
\begin{thebibliography}{99}
\footnotesize
\itemsep=-1.2mm
%\small
%\renewcommand{\baselinestretch}{1}
%\setlength{\bibsep}{0pt}
%    \setstretch{1}

\bibitem{as}M.  Ablowitz, H. Segur, Solitons and the inverse scattering transform,
%SIAM Studies in Applied Mathematics, 4. Society for Industrial and Applied Mathematics (
SIAM, Phil., PA, 1981.

\bibitem{jets}H. Baran, M. Marvan, Jets. A software for differential
calculus on jet spaces and diffieties, available online at {\tt{http://jets.math.slu.cz}}

\bibitem{bl}D.E. Blair, Riemannian geometry of contact and symplectic manifolds. 2nd ed.
Birkh\"auser, % Boston, Inc.,
Boston, MA, 2010.


\bibitem{bls}M. B\l aszak, A. Sergyeyev, Contact Lax pairs and associated (3+1)-dimensional integrable systems,
in {\em Nonlinear Systems and Their Remarkable Mathematical Structures}, vol. 2, CRC Press,
%Taylor \& Francis Group
submitted, arXiv:1901.05181.\looseness=-1

\bibitem{br} A. Bravetti,  Contact Hamiltonian dynamics: the concept and its use, Entropy 19 (2017), no. 10, art.\ 535, 12 pp.

\bibitem{ca}F. Calogero, Why are certain nonlinear PDEs both widely applicable and integrable?, in {\em What is integrability?},
%Springer Ser. Nonlinear Dynam.,
Springer, Berlin, 1991, 1--62.

\bibitem{dsk}A. De Sole, V.G. Kac, D. Valeri, A new scheme of integrability for (bi)Hamiltonian PDE,
Comm. Math. Phys. 347 (2016), no. 2, 449--488, arXiv:1508.02549.

\bibitem{df} S. Dimas, I.L. Freire,
Study of a fifth order PDE using symmetries,
Appl. Math. Lett. 69 (2017), 121--125.\looseness=-1

\bibitem{d}B. Dorizzi, B. Grammaticos, A. Ramani, P. Winternitz, Are all the equations of the Kadomtsev--Petviashvili hierarchy integrable? J. Math. Phys. 27 (1986), no. 12, 2848--2852.

\bibitem{dfk}M.  Dunajski, E.V. Ferapontov, B. Kruglikov, On the Einstein-Weyl and conformal self-duality equations. J. Math. Phys. 56 (2015), no. 8, 083501, 10 pp.

\bibitem{fer}E.V. Ferapontov, A. Moro, V.V. Sokolov,
Hamiltonian systems of hydrodynamic type in $2+1$ dimensions.
Comm. Math. Phys. 285 (2009), no. 1, 31--65.

\bibitem{hkbp}O.Ye. Hentosh, B.Yu. Kyshakevych, D. Blackmore, A.K. Prykarpatski, New fractional nonlinear integrable Hamiltonian systems, Appl. Math. Lett. 88 (2019) 41--49.

\bibitem{kod}Y. Kodama,
Dispersionless integrable systems and their solutions, in \emph{Integrability: the Seiberg--Witten and Whitham equations (Edinburgh, 1998)}, 199--212, Gordon \& Breach, Amsterdam, 2000. \looseness=-1


\bibitem{kvv}J. Krasil'shchik, A.M. Verbovetsky, R. Vitolo, The symbolic computation of integrability structures for partial differential equations, Springer, Cham, 2017.\looseness=-1

\bibitem{lr} P.-L. Lions, J.-C. Rochet, Hopf formula and multitime Hamilton-Jacobi equations, Proc. Amer. Math. Soc. 96 (1986), no. 1, 79--84.

\bibitem{mw}L.J. Mason, N.M.J. Woodhouse, Integrability, self-duality, and twistor theory, Clarendon \& OUP,
%xford Univ. Press,
N.Y., 1996.

\bibitem{o}P.J. Olver, Applications of Lie groups to differential equations,
  2nd ed., Springer, N.Y., 1993.


\bibitem{as-ro}A. Sergyeyev, A simple construction of recursion operators for multidimensional dispersionless integrable systems. J. Math. Anal. Appl. 454  (2017), no. 2, 468--480, arXiv:1501.01955.

\bibitem{scg}A. Sergyeyev, New  integrable (3+1)-dimensional systems and contact geometry, Lett. Math. Phys. 108 (2018), no. 2, 359--376, arXiv:1401.2122

\bibitem{snd}A. Sergyeyev, Integrable (3+1)-dimensional systems with rational Lax pairs, Nonlin.\ Dynamics 91 (2018), no. 3, 1677--1680, arXiv:1711.07395.\looseness=-1

\bibitem{se}D. Serre,  Syst\`emes hyperboliques riches de lois de conservation, in {\em Nonlinear partial differential equations and their applications. Coll\`ege de France Seminar, vol.\ XI (Paris, 1989--1991)}, Longman, Harlow, 1994, 248--281.

\bibitem{t}C. Taubes, What we know and what we don't know about 4-dimensions, in {\em Introduction to modern mathematics}, %Adv. Lect. Math. (ALM), 33,
   Int. Press, Somerville, MA, 2015, 391--408.

\bibitem{vps}O.O. Vaneeva, R.O. Popovych, C. Sophocleous,
Equivalence transformations in the study of integrability, Phys.\ Scr.\
89 (2014), 038003, arXiv:1308.5126.\looseness=-1

\bibitem{w}E. Witten, Searching for integrability, J. Geom. Phys. 8 (1992), no. 1-4, 327--334.

\bibitem{yan}J. Yang, Nonlinear waves in integrable and nonintegrable systems.
%Mathematical Modeling and Computation, 16. Society for Industrial and Applied Mathematics (
SIAM, Phil., PA, 2010.

\bibitem{z}V.E. Zakharov, Dispersionless limit of integrable systems in 2+1 dimensions, in {\em
Singular Limits of Dispersive Waves},
ed. by N.M. Ercolani et al.,
pp.~165--174, Plenum Press, N.Y., 1994.

\end{thebibliography}
\end{document}